\documentclass[11pt,a4paper,oneside,onecolumn]{article}

\usepackage{versions}

\excludeversion{INUTILE}
\excludeversion{LONG}
\includeversion{SHORT}

\usepackage{cite}
\usepackage[english]{babel}
\usepackage[ansinew]{inputenc}
\usepackage{epsfig,float,latexsym,enumerate,amssymb,amsmath,amsfonts,color}
\usepackage{subfigure}

\providecommand{\pb}[1]{{\sc #1} problem}
\providecommand{\entropy}{{\ensuremath\cal H}}
\providecommand{\idtt}[1]{\ensuremath{\mathtt{#1}}}
\providecommand{\nInv}{{\idtt{Inv}}}
\providecommand{\bbox}{\mathtt{Box}}
\providecommand{\xbox}{\mathtt{B}}
\providecommand{\ten}{{\idtt{Ten}}}

\newtheorem{theorem}{Theorem}[section]

\newtheorem{definition}[theorem]{Definition}

\newtheorem{lemma}[theorem]{Lemma}

\newenvironment{proof}{\begin{trivlist} \item[] {\it Proof.}}{\hfill $\Box$\end{trivlist}}

\begin{document}

\title{Adaptive Techniques to find Optimal Planar Boxes}

\author{
J. Barbay\thanks{Department of Computer Science, University of Chile, Chile}
\and
G. Navarro\thanks{Department of Computer Science, University of Chile, Chile.
Partially supported by Millennium Institute
for Cell Dynamics and Biotechnology (ICDB), Grant ICM P05-001-F, Mideplan,
Chile.}
\and
P. P\'erez-Lantero\thanks{Escuela de Ingenier\'ia Civil en Inform\'atica,
Universidad de Valpara\'{i}so, Chile. Partially supported by grant FONDECYT 11110069
and project MEC MTM2009-08652.}}

\maketitle

\begin{abstract}
  Given a set $P$ of $n$ planar points, two axes and a real-valued
  score function $f()$ on subsets of $P$, the \pb{Optimal Planar Box}
  consists in finding a box (i.e. axis-aligned rectangle) $H$
  maximizing $f(H\cap P)$.
  We consider the case where $f()$ is monotone decomposable,
  i.e. there exists a composition function $g()$ monotone in its two
  arguments such that $f(A)=g(f(A_1),f(A_2))$ for every subset
  $A\subseteq P$ and every partition $\{A_1,A_2\}$ of $A$.
  In this context we propose a solution for the \pb{Optimal Planar
    Box} which performs in the worst case $O(n^2\lg n)$ score
  compositions and coordinate comparisons, and much less on other
  classes of instances defined by various measures of difficulty.
  A side result of its own interest is a fully dynamic \textit{MCS
    Splay tree} data structure supporting insertions and deletions
  with the \emph{dynamic finger} property, improving upon previous
  results [Cort\'es et al., J.Alg. 2009].
\end{abstract}

\begin{LONG}
  \begin{minipage}[c]{.7\linewidth}
  \end{minipage}
\end{LONG}

\section{Introduction}\label{section:intro}

Consider a set $P$ of $n$ planar points, and two axes $x$ and $y$
forming a base of the plane, such that the points are in general
position (i.e. no pair of points share the same $x$ or $y$
coordinate).
We say that a real-valued function $f()$ on subsets of $P$ is
\emph{decomposable}~\cite{bautista2011,dobkin1989} if there exists a
\emph{composition function} $g()$ such that $f(A)=g(f(A_1),f(A_2))$
for every subset $A\subseteq P$ and every partition
$\{A_1,A_2\}$ of $A$.
Without loss of generality, we extend $f()$ to $P$ such that
$f(p)=f(\{p\})$.
A decomposable function is \emph{monotone} if the corresponding
composition function $g()$ is monotone in its two arguments.
A \emph{box} is a rectangle aligned to the axes, and given a monotone
decomposable function $f()$, such a box is $f()$-optimal if it
optimizes $f(H\cap P)$.
Without loss of generality, we assume that we want to maximize $f()$
and that its composition function $g()$ is monotone increasing in its
two arguments.
Given a monotone decomposable function $f()$ well defined for the
empty set $\varnothing$, a point $p$ of $P$ is \emph{positive} if
$f(p)>f(\varnothing)$. Otherwise, this point $p$ is {\em negative}.
Observe that if $p$ is positive then
$f(A\cup\{p\})=g(f(A),f(p))>g(f(A),f(\varnothing))=f(A)$ by
monotonicity of $g()$:
hence a point $p$ is positive if and only if $f(A\cup\{p\})>f(A)$ for
every subset $A\subset P$ not containing $p$.
A \emph{stripe} is an area delimited by two lines parallel to the same axis.
A \emph{positive stripe} (resp. \emph{negative stripe}) is one which contains only
positive (resp. negative) points.
A \emph{monochromatic stripe} is a stripe in which all points have the same sign.

Given a set of planar points, a simple example of such monotone
decomposable functions is counting the number of points contained in
the box.
Further examples include
counting the number of blue points;
returning the difference between the number of blue points and the
number of red points contained;
returning the number of blue points in the box or $-\infty$ if it
contains some red points;
summing the weights of the points contained;
taking the maximum of the weights of contained points; etc.

Given a set $P$ of $n$ planar points and a real-valued function $f()$
on subsets of $P$, the \pb{Optimal Planar Box} consists in finding an
$f()$-optimal box.
Depending of the choice of the function $f()$, this geometric optimization
problem has various practical applications, from identifying
rectangular areas of interest in astronomical pictures to the design
of optimal rectangular forest cuts or the analysis of medical
radiographies.
We present various adaptive techniques for the \pb{Optimal Planar
  Box}:
\begin{itemize}

\item In the worst case over instances composed of $n$ points, our
  algorithm properly generalizes Cort\'es et al.'s
  solution~\cite{cortes2009} for the \pb{Maximum Weight Box}, within
  the same complexity of $O(n^2\lg n)$ score compositions.

\item For any $\delta\in[1..n]$ and $n_1,\ldots,n_\delta\in[1..n]$
  summing to $n$, in the worst case over instances composed of
  $\delta$ monochromatic stripes of alternating signs when the points
  are sorted by their $y$-coordinates, such that the $i$-th stripe
  contains $n_i$ points, our algorithm executes $O(\delta
  n(1+\entropy(n_1,\ldots,n_\delta)))$ $\subset O(\delta
  n\lg(\delta+1))$ score compositions (Theorem~\ref{teo:adv-f()}),
  where $\entropy(n_1,\dots,n_\delta)=\sum_{i=1}^\delta (n_i/n)\lg
  (n/n_i)$ is the usual entropy function.
  \begin{LONG}
    (Of course the $x$ and $y$ axes can be interchanged.)
  \end{LONG}

\item Assuming the same $y$-coordinate order, for any
  $\lambda\in[0..n^2]$, in the worst case over instances where
  $\lambda$ is the sum of the distances between the insertion
  positions of the consecutive points according to their
  $x$-coordinate, our algorithm executes $O(n^2(1+\lg(1+\lambda/n))$
  score compositions (Lemma~\ref{lem:lambda}).
  Measure $\lambda$ relates to the local insertion sort complexity
  \cite{measuresOfPresortednessAndOptimalSortingAlgorithms} of the
  sequence of $x$-coordinates.
  It holds $\lambda \in O(n+\nInv)$, where $\nInv$ is the number of
  disordered pairs in the sequence.
  When the points are grouped into $\delta$ monochromatic stripes,
  this complexity drops to $O(n\delta(1+\lg(1+\nInv/n))$
  (Theorem~\ref{teo:rel-pos}).

\item Assuming the same $y$-coordinate order, for a minimal cover of
  the same sequence of $x$-coordinates into $\rho \le n$ \emph{runs}
  (i.e. contiguous increasing subsequences) of lengths
  $r_1,\ldots,r_\rho$, our algorithm executes
  $O(n^2(1+\entropy(r_1,\ldots,r_\rho)))$ $\subset O(n^2\lg(\rho+1))$
  score compositions (Lemma~\ref{lemma:rho}).
  When the points can be grouped into $\delta$ monochromatic stripes,
  this complexity decreases to
  $O(n\delta(1+\entropy(r_1,\ldots,r_\rho))) \subset
  O(n\delta\lg(\rho+1))$ (Theorem~\ref{teo:rel-pos} again).

\item In the case where subsets of points are clustered along the
  diagonal, our algorithm reduces to the corresponding sub-instances
  in linear time via a clever partitioning strategy
  (Theorem~\ref{teo:comp-optimal-box-diag-tree}), applying previous
  techniques on those sub-instances directly or indirectly
  (Theorem~\ref{teo:comp-optimal-box-windmills}).

\end{itemize}

\begin{LONG}
  Each of these general results applies to more specific problems such
  as the \pb{Maximum Weight Box}, where scores (i.e. the weights) can
  be composed (i.e. summed) through $g()$ in constant time.

  \medskip

\end{LONG}

We describe our algorithm progressively as follows:
\begin{enumerate}

\item After describing the existing solutions to more specific
  problems, we present a solution for the \pb{Optimal Planar Box}
  performing $O(n\lg n)$ coordinate comparisons and $O(n^2\lg n)$
  score compositions in the worst case over instances composed of $n$
  points, in Section~\ref{sec:optim-boxes-relat}.

\item We describe a truly dynamic \textit{MCS tree} data structure,
  supporting in good amortized time the insertion and deletion of
  points (as opposed to their mere activation, as in the original
  variant described by Cort\'es et al.~\cite{cortes2009}) and of
  sequences of points ordered by their coordinates along one axis,
  such as to support in particular the \emph{Dynamic Finger property}
  on score compositions and comparisons, in
  Section~\ref{sec:fully-dynamic-mcs}.

\item Observing that an instance composed mostly of positive or
  negative points is easier than a general instance, we describe a
  technique to detect and take advantage of monochromatic stripes,
  hence taking advantage of the sign distribution of the values of
  $f()$, in Section~\ref{sec:taking-advantage-f}.

\item Observing that the problem can be solved in $O(n\lg n)$
  coordinate comparisons and $O(n)$ score compositions when the points
  are projected in the same order on both axis, we describe in
  Section~\ref{sec:taking-advant-relat} a variant of the algorithm
  which takes advantage of the relative insertion positions of the
  consecutive points into a sorted sequence, by using our new dynamic
  MCS tree data structure.
  The resulting cost is related to that of local insertion sorting~\cite{measuresOfPresortednessAndOptimalSortingAlgorithms,estivillcastro92survey,anOverviewOfAdaptiveSorting},
  and we relate it with other measures of disorder of permutations,
  such as the number of inversions and the number of consecutive
  increasing
  runs~\cite{estivillcastro92survey,anOverviewOfAdaptiveSorting}.

\item Observing that sorting the points within each monochromatic
  stripe can only further reduce the inversions and runs complexity
  between the orders of the points, we combine the two previous
  techniques into a single algorithm, whose complexity improves upon
  both techniques, in Section~\ref{sec:final-algorithm}.

\item Observing that one dimensional instances can be solved in linear
  time even when plunged in the plane, whereas the techniques
  described above have quadratic complexity on such instances, we
  describe how to decompose instances where the points are even only
  partially aligned into diagonals of blocks, which yields an
  algorithm whose best complexity is linear and degrades smoothly to
  the previous worst case complexity as the points lose their
  alignment.

\item Finally, observing that sub-instances which cannot be decomposed
  into diagonals of blocks have a specific ``windmill'' shape, we show
  how to take advantage of this shape in order to still take advantage
  of particular instances.

\end{enumerate}

The proofs of the most technical results are presented in the
appendix, following an exact (automatic) copy of the statement of the
corresponding result.

\section{Optimal Boxes and Related Problems}
\label{sec:optim-boxes-relat}

\begin{LONG}
  \subsection{Maximization problems on rectangles in the plane}
  \label{sec:maxim-probl-rect}
\end{LONG}

Given a set $P$ of $n$ weighted planar points, in which the weight of a point can be
either positive or negative, the \pb{Maximum Weight Box}~\cite{cortes2009} consists in
finding a box $R$ maximizing the sum of the weights of the points in $R\cap P$.
Cort\'es et al.~\cite{cortes2009} gave an algorithm solving this problem in time
$O(n^2\lg n)$ using $O(n)$ space, based on \emph{MCS trees}, a data structure supporting
in $O(\lg n)$ time the dynamic computation of the \pb{Maximum-Sum Consecutive
Subsequence}~\cite{bentley1984} (hence the name ``MCS'').
Their solution applies to other similar problems, such as the \pb{Maximum
Subarray}~\cite{takaoka2002} which, given a two-dimensional array, consists in finding a
subarray maximizing the sum of its elements.

\begin{LONG}
  Given a set $P$ of $n$ planar points, each being colored either blue
  or red, the \pb{Maximum Box}~\cite{cortes2009,eckstein2002,liu2003}
  consists in finding a box $R$ containing the maximum number of blue
  points and no red point.
  First solved in $O(n^2\lg n)$ time by Liu and Nediak~\cite{liu2003},
  the problem was later reduced to a particular case of the
  \pb{Maximum Weight Box}~\cite{cortes2009}, in which blue points have
  weight $+1$ and red points have weight $-\infty$, and was thus
  solved within the same time complexity of $O(n^2\lg n)$ but with a
  different and simpler approach.
  Recently, Backer et al.~\cite{backer2010} showed that the
  \pb{Maximum Box} can be solved in $O(n\lg^3n)$ time and $O(n\lg n)$
  space. For dimension $d\geq 3$, the \pb{Maximum Box} in
  $\mathbb{R}^d$ is NP-hard if $d$ is part of the input of the
  problem~\cite{eckstein2002}. For fixed values of $d\geq 3$, the
  \pb{Maximum Box} in $\mathbb{R}^d$ can be solved in $O(n^d\lg
  ^{d-2}n)$ time~\cite{backer2010}.
\end{LONG}

\begin{LONG}
Given a set $P$ of $n$ planar points, each being colored either blue or red, the
\pb{Maximum Discrepancy Box}~\cite{cortes2009,dobkin1996} consists in finding a box that
maximizes the absolute difference between the numbers of red and blue points that it
contains.
Dobkin et al.~\cite{dobkin1996} gave a $O(n^2\lg n)$-time algorithm for it, after which
the problem was reduced to two particular cases of the \pb{Maximum Weight
Box}~\cite{cortes2009} with the same time complexity.
The first case is when red points have weight $+1$ and blue points weight $-1$, and vice
versa in the second case. We should point out that neither of the above algorithms to find
maximum boxes takes advantage of the distribution of the point set $P$ to reduce the time
complexity.
\end{LONG}

\begin{LONG}
Bereg et al.~\cite{bereg2010} solved the \pb{Moving Points Maximum Box}, a generalization
of the \pb{Maximum Box} to the context of moving points with an $O(n^2)$-space dynamic
data structure that maintains the solution up to date while the relative positions of the
points changes.
This yields an $O(n^3\lg n)$-time $O(n^2)$-space algorithm to solve the \pb{Maximum
Rectangle}, where the rectangle is not restricted to be axis aligned.
\end{LONG}

\begin{LONG}
  As mentioned in the introduction, the \pb{Optimal Planar Box}
  consists in finding a box (i.e. axis-aligned rectangle) $H$
  maximizing $f(H\cap P)$, given a set $P$ of $n$ planar points, two
  axes and a real-valued monotone decomposable score function $f()$ on
  subsets of $P$.
\end{LONG}
The \pb{Maximum Weight Box}~\cite{cortes2009} and, by successive
reductions, the \pb{Maximum Subarray}~\cite{takaoka2002}, the
\pb{Maximum Box}~\cite{cortes2009,eckstein2002,liu2003}, and the
\pb{Maximum Discrepancy Box}~\cite{cortes2009,dobkin1996} can all be
reduced to a finite number of instances of the \pb{Optimal Planar Box}
by choosing adequate definitions for the score functions $f()$ to
optimize.

\begin{LONG}
  \subsection{Algorithm for Optimal Boxes}
  \label{sec:algor-optim-boxes}

  The basic $O(n^2\lg n)$ algorithm~\cite{cortes2009} for the
  \pb{Maximum Weight Box} can easily be modified in order to handle
  increasing decomposable functions.
\end{LONG}

Cort\'es et al.'s algorithm~\cite{cortes2009} first sorts
the points by their $y$-coordinate in $O(n\lg n)$ time and then
traverses the resulting sequence of points $p_1, p_2, \ldots p_n$ as
follows.
For each $p_i$, it sets an MCS tree (described in more details in
Section~\ref{sec:fully-dynamic-mcs}) with points $p_i, \ldots p_n$,
where the key is their $x$-coordinate $x_i$, and all have value
$f(\varnothing)$.
It then successively activates points $p_j$ for $j\in[i..n]$, setting
its weight to value $f(p_j)$, updating the MCS tree so that to compute
the optimal box contained between the $y$-coordinate of $p_i$ to that
of $p_j$.
\begin{INUTILE}
  Points are sorted twice, once by $x$ and once by $y$ coordinate.
  Each point is augmented with a reference to the corresponding leaf.
\end{INUTILE}
The whole algorithm executes in time $O(n^2\lg n)$, corresponding to
$n^2$ activations in the tree, each performed in time $O(\log n)$.

\section{Fully Dynamic MCS Trees} 
\label{sec:fully-dynamic-mcs}

Cort\'es et al.~\cite{cortes2009} defined the MCS tree as an index for
a fixed sequence $S=(x_i)_{i\in[1..n]}$ of $n$ elements, where each
element $x_k$ of $S$ has a weight $w(x_k)\in\mathbb{R}$, so that
whenever a weight $w(x_k)$ is updated, a consecutive subsequence
$(x_i)_{i\in[l..r]}$ of $S$ maximizing $\sum_{i\in[l..r]} w(x_i)$ is
obtained (or recomputed) in $O(\lg n)$ time.
\begin{LONG}
  This data structure can adapt to the deactivation or activation of
  any element of the sequence, where deactivation of an element means
  to assign weight zero to it, and activation means to assign to a
  deactivated element its former weight.
\end{LONG}
This behavior is dynamic in the sense that it allows modification of
element weights, yet it is only partially dynamic in the sense that it
admits neither addition of new elements nor deletion of existing
elements.


Existing dynamic data structures can be easily adapted into a truly
dynamic data structure with the same functionalities as MCS trees.
We start by generalizing Cort\'es et al.'s algorithm and
data-structure~\cite{cortes2009} from mere additive weights to
monotone decomposable score functions in
Lemma~\ref{lemma:mcs-static-tree}.
We further generalize this solution to use an AVL tree~\cite{avl1962}
in Lemma~\ref{lemma:mcs-avl-tree} and a Splay tree~\cite{sleator1985}
in Lemma~\ref{lemma:mcs-splay-tree}, whose ``finger search'' property
will yield the results of Sections~\ref{sec:taking-advantage-f}
and~\ref{sec:taking-advant-relat}.

\begin{LONG}
  \subsection{MCS Tree for Monotone Decomposable Score Functions}
  \label{sec:overview-MCS-tree}

  Our first and simplest result is to generalize Cort\'es et al.'s
  algorithm and data-structure~\cite{cortes2009} from mere additive
  weights to monotone decomposable score functions.
\end{LONG}

\begin{lemma}\label{lemma:mcs-static-tree}
  Let $S$ be a static sequence of $n$ elements,
  and $f()$ be a monotone decomposable score function receiving as
  argument any subsequence of $S$, defined through the activation and
  deactivation of each element of $S$.
  There exists a semi-dynamic data structure for maintaining $S$
  using linear space that supports
  the search for an element in $O(\lg n)$ comparisons;
  the activation or deactivation of an element in $O(\lg n)$ score
  compositions; and
  $f()$-optimal sub range queries in $O(\lg n)$ comparisons and score
  compositions.
\end{lemma}

\begin{proof}
  Consider a monotone decomposable score function $f()$ on sequences
  (a particular case of monotone decomposable score functions on point
  sets described in Section~\ref{section:intro}) and its corresponding
  composition function $g()$.

  Given a fixed sequence $S=(x_i)_{i\in[1..n]}$ of $n$ elements, each
  element $x_i$ of $S$ is associated to a score $f(x_i)$ (a more
  general notion than the weight previously considered, in that
  $f(x_1,\ldots,x_n)$ can be computed more generally than by merely
  the sum of the individual scores $f(x_i)$).

  Then the MCS tree data structure consists of a binary balanced tree
  $T$ with $n$ leaves, in which the leaves of $T$ from left to right
  store the elements $x_1,x_2,\ldots,x_n$ of $S$.
  The term {\em interval} is used to denote the first and last indices
  of a consecutive subsequence of $S$ joint with the score of the
  corresponding sequence.
  We denote by $[x_l,x_r]$ the interval corresponding to the
  subsequence $(x_i)_{i\in[l..r]}$, and $[x_k,x_k]$ by $[x_k]$.
  Let $\varnothing$ denote the empty interval whose score is equal to
  zero, and given intervals $I_1,\ldots,I_t$ let
  $\max\{I_1,\ldots,I_t\}$ be the interval of maximum score among
  $I_1,\ldots,I_t$.
  Each node $v$ of $T$ stores the following four maximum-score
  intervals, where $x_l,x_{l+1},\ldots,x_r$ are, from left to right,
  the leaves of $T$ descendant of $v$:
  \begin{itemize}
  \item $I(v)=[x_l,x_r]$, the information about the full sequence;
  \item
    $L(v)=\max\{[x_l,x_l],[x_l,x_{l+1}],\ldots,[x_l,x_r],\varnothing\}$,
    the information about the best prefix of the sequence (where $L$
    stands for ``Left'');
  \item
    $R(v)=\max\{[x_l,x_r],[x_{l+1},x_r],\ldots,[x_r,x_r],\varnothing\}$,
    the information about the best suffix of the sequence (where $R$
    stands for ``Right''); and finally
  \item $M(v)=\max\{\max\{[x_{l'},x_{r'}]:l\leq l'\leq r'\leq
    r\},\varnothing\}$, the information about the best factor of the
    sequence (where $M$ stands for ``Middle'').
  \end{itemize}
  Then the maximum-score consecutive subsequence of $S$ is given by
  the interval $M()$ of the root node of $T$.
  Given two contiguous intervals $[x_l,x_k]$ and $[x_{k+1},x_r]$, let
  $[x_l,x_k]+[x_{k+1},x_r]$ denote the interval $[x_l,x_r]$, of score
  obtained through the combination of the scores of of $[x_l,x_k]$ and
  $[x_{k+1},x_r]$.
  Consider $I+\varnothing=\varnothing+I=I$ for every interval $I$.
  The key observation in the MCS tree is that if we have computed
  $I(v_1)$, $L(v_1)$, $R(v_1)$, $M(v_1)$, $I(v_2)$, $L(v_2)$,
  $R(v_2)$, and $M(v_2)$, where nodes $v_1$ and $v_2$ are the left and
  right children of node $v$, respectively, then $I(v)$, $L(v)$,
  $R(v)$, and $M(v)$ can be computed through a constant number of
  operations:
  $I(v)=I(v_1)+I(v_2)$,
  $L(v)=\max\{L(v_1),g(I(v_1),L(v_2))\}$,
  $R(v)=\max\{g(R(v_1),I(v_2)),R(v_2)\}$, and
  $M(v)=\max\{M(v_1),M(v_2),g(R(v_1),L(v_2))\}$.
  These observations can be used, whenever the score of an element
  $x_i$ of $S$ is modified, to update through a constant number of
  operation each node in the path from the leaf of $T$ storing the
  score of a sequence containing $x_i$ to the root.
  Since $T$ is a balanced tree this corresponds to $O(\lg n)$
  operations.
  The MCS tree data structure also supports optimal subrange queries
  through $O(\lg n)$ operations, that is, given a range $\mathtt{R}$
  of sequence $S$ reports a subrange $\mathtt{R}'$ of $\mathtt{R}$
  such that the score of $S$ in $\mathtt{R}'$ is maximized.
\end{proof}

\begin{LONG}
  \subsection{MCS AVL Tree and Dynamic Insertion and Deletion}
  \label{sec:MCS-AVL-tree}
\end{LONG}

The MCS tree data structure can be converted into a truly dynamic data
structure supporting both insertions and deletions of elements.
This data structure can be used to index a dynamic sequence
$S=(x_i)_{i\in[1..n]}$ of $n$ elements so that whenever an element is
inserted or removed, a consecutive subsequence $S'=(x_i)_{i\in[l..r]}$
of $S$ optimizing $f(S')$ can be (re)computed in $O(\lg n)$ score
compositions and comparisons.
The following lemma establishes the property of this data structure,
which we call {\em MCS AVL tree}.

\begin{lemma}\label{lemma:mcs-avl-tree}
  Let $S$ be a dynamic sequence of $n$ elements,
  and $f()$ be a monotone decomposable score function receiving as
  argument any consecutive subsequence of $S$.
  There exists a fully dynamic data structure for maintaining $S$
  using linear space that supports
  the search for an element in $O(\lg n)$ comparisons;
  the update of the weight of a point in $O(\lg n)$ score
  compositions,
  the insertion or deletion of an element in $O(\lg n)$ comparisons
  and score compositions; and
  $f()$-optimal subrange queries in $O(\lg n)$ comparisons and score
  compositions.
\end{lemma}

\begin{proof}
  Let $S=(x_i)_{i\in[1..n]}$, $f()$ the monotone decomposable score
  function, and $g()$ its composition function.
  Consider an AVL tree $T$ of $n$ nodes such that for $i\in[1..n]$ the
  $i$-th node in the in-order traversal of $T$ stores element $x_i$ of $S$.
  We now generalize the MCS data structure.

  Each node $v$ of $T$, in which the in-order traversal of the subtree rooted at $v$ reports the elements $x_l,x_{l+1},\ldots,x_r$ of $S$, is augmented with the same four intervals $I(v)$, $L(v)$, $R(v)$, and $M(v)$ used by the MCS data structure.
  If node $v$ stores the element $x_k$ and if the intervals $I(v_1)$, $L(v_1)$, $R(v_1)$, $M(v_1)$, $I(v_2)$, $L(v_2)$, $R(v_2)$, and $M(v_2)$ of the left ($v_1$) and right ($v_2$) children of $v$ have been computed; then $I(v)$, $L(v)$, $R(v)$, and $M(v)$ can be computed in constant number of score compositions as follows:
  \begin{itemize}
  \item $I(v)=I(v_1)+[x_k]+I(v_2)$
  \item $L(v)=\max\{L(v_1),I(v_1)+[x_k]+L(v_2)\}$
  \item $R(v)=\max\{R(v_1)+[x_k]+I(v_2),R(v_2)\}$
  \item $M(v)=\max\{M(v_1),M(v_2),R(v_1)+[x_k]+L(v_2)\}$
  \end{itemize}
  In this computation the value $f()$ of every interval of the form $I_1+[x_k]+I_2$ is equal to $g(g(f(S_1),f(x_k)),f(S_2))$, where $S_1$ and $S_2$ are the subsequences corresponding to $I_1$ and $I_2$, respectively.
  For empty (or null) nodes $v$, $I(v)=L(v)=R(v)=M(v)=\varnothing$.
  This states how these intervals are computed for leaves and one-child nodes.
  We show that the computation of $L(v)$ is correct by considering separately the case where $L(v)$ contains $x_k$ and the case where it does not.
  If $x_k$ is not in $L(v)$, then $L(v)$ must consider only the elements $x_l,x_{l+1},\ldots,x_{k-1}$, and is thus equal to $L(v_1)$.
  Otherwise, if $x_k$ belongs to $L(v)$, then $L(v)$ must have the form $I(v_1)+[x_k]+I_2$, where $I_2$ is either $\varnothing$ or an interval of the form $[x_{k+1},x_j]$, $j\in[k+1..r]$.
  Observe that $I(v_2)$ has the same form as $I_2$ and maximizes $f()$.
  Since $g()$ is monotone increasing in its first and second arguments, then
  \begin{eqnarray*}
    f(I(v_1)+[x_k]+I_2)&=&g\left(g(f(I(v_1)),f(x_k)),f(I_2)\right)\\
    &\leq&g(g(f(I(v_1)),f(x_k)),f(I(v_2)))\\
    &=&f(I(v_1)+[x_k]+I(v_2)).
  \end{eqnarray*}
  Therefore, $I(v_1)+[x_k]+I(v_2)$ is a correct choice for $L(v)$ if $x_k$ belongs to
  $L(v)$. Similar arguments can be given for $R(v)$ and $M(v)$.
  
  If the elements of $S$ are ordered by some key, then $T$ is a binary search tree on that key, and insertions and deletions are performed by using the binary search criterion. Otherwise, if elements of $S$ are not ordered, we further augment every node $v$ of $T$ with the number of nodes in the subtree rooted at $v$ in order to perform both insertions and deletions by rank~\cite{cormen2001}.
  In either case, whenever we insert or delete an element we spend $O(\lg n)$ comparisons.  Furthermore, when a leaf node is inserted, or a leaf or one-child node is deleted, we perform a bottom-up update of all nodes of $T$ in the path from that node to the root.
  Let $v$ be any of those nodes.
  If the update at $v$ does not require a rotation, then the augmented information of $v$ is computed in $O(1)$ score compositions from the same information of its two children nodes, which are already computed.
  Otherwise, if the update $v$ requires a rotation, once the rotation is performed, we update bottom-up the augmented information of the nodes changed in the rotation (they include $v$ and the node getting the position of $v$) from the same information of their new children nodes.
  This is done again in $O(1)$ score compositions since there are at most three affected nodes in any rotation and the augmented information of their new children nodes is already computed.
%
  The update of $v$ always requires $O(1)$ score compositions, and then $O(\lg n)$ score compositions are used in total at each insertion or deletion.
  The subsequence $S'$ corresponding to the interval $M()$ of the root of $T$ is the subsequence of $S$ that maximizes $f()$, and is updated at each insertion or deletion.

  Optimal subrange queries on $T$ can be performed as follows. Let $x_l$ and $x_r$ be the first and last elements of the input range $\mathtt{R}$, respectively. Using $O(\lg n)$ comparisons we find: the node $v_l$ storing $x_l$, the node $v_r$ storing $x_r$, and the least common ancestor node $v_{l,r}$ of $v_l$ and $v_r$.
  Then the intervals $R_l$, $M_l$, $L_r$, and $M_r$ are computed: We first initialize $R_l:=M_l:=\varnothing$, and after that, for each node $v$ taken bottom-up from $v_l$ to the left child of $v_{l,r}$, such that element $x_k$ stored in $v$ belongs to $\mathtt{R}$, we perform the updates $R_l:=\max\{R_l+[x_k]+I(v_2),R(v_2)\}$ and $M_l:=\max\{M_l,M(v_2),R_l+[x_k]+L(v_2)\}$ in this order, where $v_2$ is the right child of $v$. Similarly, we initialize $R_r:=M_r:=\varnothing$ and for each node $v$ taken bottom-up from $v_r$ to the right child of $v_{l,r}$, such that element $x_k$ stored in $v$ belongs to $\mathtt{R}$, we do $L_r:=\max\{L(v_1),I(v_1)+[x_k]+L_r\}$ and $M_r:=\max\{M(v_1),M_r,R(v_1)+[x_k]+L_r\}$, where $v_1$ is the left child of $v$. Finally, the optimal subrange of $\mathtt{R}$ is equal to $\max\{M_l,M_r,R_l+[x_{l,r}]+L_r\}$, where $x_{l,r}$ is the element of $S$ stored in $v_{l,r}$. The whole computation requires $O(\lg n)$ score compositions.

Therefore, the tree $T$
satisfies all conditions of the lemma.
\end{proof}

\begin{LONG}
  \subsection{MCS Splay Tree and Dynamic Finger Property}
  \label{sec:MCS-Splay-tree}
\end{LONG}

The Splay tree is a self-adjusting binary search tree created by
Sleator and Tarjan~\cite{sleator1985}. 
It supports the basic operations search, insert and delete, all of
them called {\em accesses}, in $O(\lg n)$ amortized time. 
For many sequences of accesses, splay trees perform better than other
search trees, even when the specific pattern of the sequences are
unknown.
\begin{LONG}
  It is a simply binary search tree in which each access can cause
  some rotations in the tree, similar but more diverse than the
  rotations used by the AVL tree, so that each element accessed is
  moved to the root by a succession of such rotations.
\end{LONG}
Among other properties of Splay trees, we are particularly interested
in the {\em Dynamic Finger Property}, conjectured by Sleator and
Tarjan~\cite{sleator1985} and proved by Cole et al.~\cite{cole2000-1}:
every sequence of $m$ accesses on an arbitrary $n$-node Splay tree
costs $O(m+n+\sum_{j=1}^m \lg(d_j+1))$ rotations where, for $j=1..m$,
the $j$-th and $(j-1)$-th accesses are performed on elements whose
ranks among the elements stored in the Splay tree differ by $d_j$.
For $j=0$, the $j$-th element is the element stored at the root. 
It is easy to see that in the MCS AVL tree we can replace the
underlying AVL tree by a Splay tree, and obtain then the next lemma,
which describes the {\em MCS Splay tree} data structure.

\begin{lemma}\label{lemma:mcs-splay-tree}
  Let $S$ be a dynamic sequence of $n$ elements and $f()$ be a
  monotone decomposable function receiving as argument any consecutive
  subsequence of $S$.
  There exists a data structure for maintaining $S$ that uses linear
  space and supports 
  the search in $O(\lg n)$ amortized comparisons,
  the update of the weight of a point in $O(\lg n)$ amortized score
  compositions,
  and the insertion and deletion of elements in $O(\lg n)$ amortized
  comparisons and score compositions.
  Joint with the insertion or deletion of any element, the consecutive
  subsequence $S'$ of $S$ maximizing $f(S')$ is recomputed.
  The Dynamic Finger Property is also satisfied for each operation
  (search, insertion and deletion), both for the number of comparisons
  and for the number of score compositions performed.
\end{lemma}

\begin{proof}
  The complexities of the accesses, and also the Dynamic Finger
  Property, follow from the Splay tree properties and from the fact
  that the augmented information in each node can be computed in a
  constant number of score compositions and comparisons, from the
  corresponding information stored in its children nodes.
  Since after each rotation the augmented information of the affected
  nodes of the rotation is updated, the consecutive subsequence $S'$
  of $S$ maximizing $f(S')$, which is stored at the root, is
  recomputed.

  As in the case of the MCS AVL tree, updating the weight of an
  element can be reduced to removing and inserting it with its new
  weight, so that its support in $O(\lg n)$ amortized comparisons and
  score compositions is a simple consequence of the support of
  insertions and deletions.
  Obviously, in practice the update operator can be implemented much
  more simply but still requires a rebalancing of the tree on the
  element accessed in order to yield the amortized complexity.
\end{proof}

\begin{LONG}
Among all the instances possible over $n$ points, some instances are easier than some
others, and using the MCS Splay tree instead of a static balanced binary tree permits to
take advantage of some of those instances.
We describe in the next section the simplest way to do so, by taking advantage of the
proportion of points contributing positively to the score of boxes.
\end{LONG}

\section{Taking Advantage of Monochromatic Stripes} 
\label{sec:taking-advantage-f}

Consider an instance %
\begin{LONG}
  where all the points contributing positively to the score of boxes
  (the positive points) are gathered in the top half of the plane,
  while all the points contributing negatively (the negative points)
  are gathered in the bottom of the plane: trivially, one just needs
  to compute a box containing all the positive points and no negative
  points to solve the \pb{Optimal Planar Box}, a task performed in
  linear number of comparisons and score compositions.

  Now consider a more general instance,
\end{LONG} %
where positive and negative points can be clustered into $\delta$
positive and negative stripes along one given axis, of cardinalities
$n_1,\ldots,n_\delta$. On such instances one does not need to consider
boxes whose borders are in the middle of such stripes: all optimal
boxes will start at the edge of a stripe; specifically, the top
(resp. bottom) of an optimal box will align with a positive point at
the top (resp. bottom) of a positive stripe.

This very simple observation not only limits the number of boxes for
which we need to compute a score, but also it makes it easier to
compute the score of each box: adding the $n_i$ points of the $i$-th
stripe in increasing order of their coordinates in a MCS Splay tree of
final size $n$ amortizes to $O(n+\sum_{i=1}^\delta n_i\lg(n/n_i))$
coordinate comparisons and score compositions.
The reason is that the $n_i$ distances $d_j+1$ of
Lemma~\ref{lemma:mcs-splay-tree} telescope to at most $n+n_i$ within
stripe $i$, and thus by convexity the cost $O(n+\sum_{j=1}^n
\lg(d_j+1))$ is upper bounded by
\begin{eqnarray*}
O\left(n+\sum_{i=1}^\delta n_i\lg(1+n/n_i)\right)&\subset&O\left(n+\sum_{i=1}^\delta
n_i\lg(n/n_i)\right)\\
&=&O(n(1+\entropy(n_1,\ldots,n_\delta)))\\
&\subset&O(n\lg(\delta+1))
\end{eqnarray*}
%
Combining this with the fact that the top of an optimal box is aligned
with a positive point at the top of a positive stripe yields the
following result.

\begin{theorem}\label{teo:adv-f()}
  For any $\delta\in[1..n]$ and $n_1,\ldots,n_\delta\in[1..n]$ summing to $n$,
  in the worst case over instances composed of $\delta$ stripes of
  alternating signs over an axis such that the $i$-th stripe contains
  $n_i$ points,
  there exists an algorithm that finds an $f()$-optimal box in
  $O(\delta n(1+\entropy(n_1,\ldots,n_\delta)))\subset O(\delta
  n\lg(\delta+1))$ score compositions and $O(\delta
  n(1+\entropy(n_1,\ldots,n_\delta))+n\lg n)\subset O(\delta
  n\lg(\delta+1)+n\lg n)$ coordinate comparisons.
\end{theorem}

\begin{LONG}
  The gradient analysis of images where only a central figure
  generates positive points yields instances where negative points
  cluster into stripes.
  More generally, any instance with an imbalance between the number of
  positive and negative points will naturally generate such
  clusters. This is captured in the next corollary.

\stmt{corr}{cor:adv-f()}
  For any $p\in[0..n]$, in the worst case over instances composed of
  $p$ positive points and $n{-}p$ negative points,
  there exists an algorithm that finds an $f()$-optimal box in
  $O(n\delta\lg(\delta+1))$ score compositions, where $\delta \le
  1+2\min(p,n-p)$.
}

  In particular, this last corollary implies an algorithm running in
  $O(pn\lg p)$ score compositions for the \pb{Optimal Planar Box} when
  $p\ge 1$ points have positive (resp. negative) weight.

\end{LONG}

\section{Taking Advantage of Point Alignments} 
\label{sec:taking-advant-relat}

Running the algorithm outlined in the first paragraph of
Section~\ref{sec:taking-advantage-f} over the MCS Splay tree has
further consequences.
In this section we show how it makes the algorithm adaptive to local
point alignments.

Imagine that, once the points $p_1, p_2, \ldots, p_n$ have been sorted
by $y$-coordinate, they also become sorted by $x$-coordinate.
By hypothesis, no pair of points are aligned to any axis, so that the
second order is a permutation of the first one.
Call $\pi$ the permutation that rearranges the points $p_1, p_2,
\ldots, p_n$ from an ordering sorted by $x$-coordinates to an ordering
by $y$-coordinates.

Then, when we insert the points in the MCS Splay tree, the distance
from each insertion to the previous one is equal to 1, and therefore
the overall cost of the algorithm is $O(n^2)$ score composition
(recall the description of Cort\'es algorithm in
Section~\ref{sec:optim-boxes-relat}).

\providecommand{\insertionCostComplement}{\ensuremath{r}}

More generally, assume that the $x$-coordinate $x_j$ of $p_j$ falls at
position $\insertionCostComplement_j$ in the sorted set of previous $x$-coordinates $\{x_1,
x_2, \ldots, x_{j-1}\}$.
Then we have $d_1=0$ and $d_j = |\insertionCostComplement_j - \insertionCostComplement_{j-1}|$ for $j>1$,
according to Lemma~\ref{lemma:mcs-splay-tree}.
Let us define $\lambda = \sum_{j=2}^n |\insertionCostComplement_j-\insertionCostComplement_{j-1}|$ as the {\em
  local insertion complexity} of the sequence
\cite{measuresOfPresortednessAndOptimalSortingAlgorithms}.
Note that a simple upper bound is $\lambda \le \sum_{j=2}^n |\pi_j-\pi_{j-1}|$, as the latter refers to the final positions of the elements, whereas $\lambda$ refers to the moves needed in the current prefix of the permutation.

The cost of our algorithm using the MCS Splay tree can then be upper bounded as follows.
When we insert the points in the MCS Splay tree starting from $p_1$,
the total cost is $O(n + \sum_{j=1}^n \lg(d_j+1)) \subset O(n +
n\lg(1+\lambda/n))$ score compositions, by convexity of the logarithm
and because $\sum_{j=1}^n d_j+1 \leq \lambda+n$.
A simple upper bound when considering all the $n$ passes of the
algorithm is obtained as follows.

\begin{lemma}\label{lem:lambda}
  Let $P$ be a set of $n$ points in the plane.
  Let $f()$ be a monotone decomposable score function receiving as
  argument any subset of $P$.
  There exists an algorithm that finds an $f()$-optimal box in
  $O(n^2(1+\lg(1+\lambda/n)))$ score compositions and
  $O(n^2(1+\lg(1+\lambda/n))+n\lg n)$ coordinate comparisons, where
  $\lambda \le n^2$ is the local insertion complexity of the sequence
  of $x$-coordinates of the points sorted by $y$-coordinates.
\end{lemma}

In the worst case this boils down to the $O(n^2\lg n)$-worst-case
algorithm, whereas in the best case $\lambda=0$ and the cost
corresponds to $O(n^2)$ operations.

We can upper bound $\lambda$ by using other measures of disorder in permutations.
For example, let us consider $\nInv$, the number of {\em inversions}
in the permutation $\pi$, or said another way, the number of pairs out
of order in the sequence \cite{anOverviewOfAdaptiveSorting}.
The measure $\nInv$ corresponds to a cost where the ``finger'' is always
at the end of the sequence.
This can be as small as $(\lambda-n)/2$, for example consider the
permutation $\pi = (m,m-1,m+1,m-2,m+2,\ldots,1,2m-1)$ for $m=(n+1)/2$
and
odd $n$.
However, $\nInv$ can be much larger than $\lambda$ because it is not
symmetric on decreasing sequences, for example when the points are
semi-aligned in a decreasing diagonal and the permutation is
$\pi=(n,n-1,n-2,\ldots,1)$.
Thus replacing $\lambda$ by $\nInv$ in Lemma~\ref{lem:lambda} yields a
valid upper bound in terms of big-O complexity.

\begin{lemma}{lemma:relacionInvLambda}
For any permutation $\pi$ of
  local insertion complexity $\lambda$ and presenting $\nInv$
  inversions,
 $$\nInv \ge \lambda/2-n.$$
\end{lemma}

\begin{proof}
  For each $x_j$, $\lambda$ increases by $|\pi_j-\pi_{j-1}|$, whereas $\nInv$ increases by $j-\pi_j$. Since $\pi_j \le j$, if $\pi_j > \pi_{j-1}$, we can use that $(j-1)-\pi_{j-1}+1 \ge \pi_j-\pi_{j-1}$, or else we use $j-\pi_j \ge \pi_{j-1}-\pi_j$.
  So either the contribution of $x_{j-1}$ plus 1, or that of $x_j$, to $\nInv$, upper bounds the contribution of $x_j$ to $\lambda$.
  By summing up both contributions (that of $x_{j-1}$ plus 1, and that of $x_j$) for each $j$ we have $2\cdot\nInv+n$ and this upper bounds $\lambda$.
  See also Estivill-Castro and Wood's survey~\cite{estivillcastro92survey}. 
\end{proof} 

Another well-known measure of permutation complexity is the number of {\em increasing runs} $\rho$, that is, the minimum number of contiguous monotone increasing subsequences that cover $\pi$ \cite{theArtOfComputerProgramming}.
Let $r_1,\ldots,r_\rho$ be the lengths of the runs, computed in $O(n)$ comparisons.
Then the sum of the values $|\pi_{j+1}-\pi_j|$ within the $i$-th run
telescopes to at most $n$, and so does the sum of the $d_j$ values.
Therefore 
$\sum_{j=1}^n \lg(d_j+1) 
\le \sum_{i=1}^\rho r_i\lg(1+n/r_i) 
\le n + \sum_{i=1}^\rho r_i\lg(n/r_i)$ by convexity.
This leads to the following alternative upper bound.

\begin{lemma}\label{lemma:rho}
  Let $P$ be a set of $n$ points in the plane.
  Let $f()$ be a monotone decomposable function receiving as argument
  any subset of $P$.
  There exists an algorithm that finds an $f()$-optimal box in
  $O(n\lg n)$ coordinate comparison and
  $O(n^2(1+\entropy(r_1,\ldots,r_\rho)) \subset O(n^2\lg(\rho+1))$ score
  compositions, where
  $r_1,\ldots,r_\rho$ are the lengths of $\rho$ maximal contiguous
  increasing subsequences that cover the sequence of $x$-coordinates of
  the points sorted by $y$-coordinate.
\end{lemma}

\begin{LONG}
Of course we can consider decreasing instead of increasing runs.
\end{LONG}

\section{Taking Advantage of both Stripes and Alignments}
\label{sec:final-algorithm}

The combination of the techniques of
Sections~\ref{sec:taking-advantage-f} and
\ref{sec:taking-advant-relat} can be elegantly analyzed.
A simple result is that we need to start only from $\delta$ different
$p_i$ values, and therefore an upper bound to our complexity is
$O(n\delta((1+\lg(1+\lambda/n)))$. 
We can indeed do slightly better by sorting the points by increasing
$x$-coordinates within each monochromatic stripe.
While the measure $\lambda'$ resulting from this reordering may be larger
than $\lambda$, the upper bounds related to $\nInv$ and $\rho$, namely
$\nInv'$, $\rho'$, and $\entropy(n_1',\ldots,n_{\rho'}')$, do not increase.
In particular it is easy to see that the upper bound of
Theorem~\ref{teo:adv-f()} is dominated by the combination since $\rho'
\le \delta$ and $\entropy(r_1',\ldots,r_{\rho'}') \le
\entropy(n_1,\ldots,n_\delta)$ (because no run will cut a
monochromatic stripe once the latter is reordered).

\begin{theorem}\label{teo:rel-pos}
  Let $P$ be a set of $n$ points in the plane.
  Let $f()$ be a monotone decomposable function receiving as argument
  any subset of $P$.
  There exists an algorithm that finds an $f()$-optimal box in
  $O(n\lg n)$ coordinate comparisons and
  $O(n\delta(1+\min(\lg(1+\nInv/n),\entropy(r_1,\ldots,r_\rho))))
  \subset O(n\delta\lg(\rho+1))$ score compositions, where
  $\delta$ is the minimum number of monochromatic stripes in which the
  points, sorted by increasing $y$-coordinate, can be partitioned;
  $X$ is the corresponding sequence of $x$-coordinates once we
  (re-)sort by increasing $x$-coordinate the points within each
  monochromatic stripe;
  $\nInv \le n^2$ is the number of out-of-order pairs in $X$; and
  $r_1,\ldots,r_\rho$ are the lengths of the minimum number
  $\rho\le\delta$ of contiguous increasing runs that cover $X$.
  A similar result holds by exchanging $x$ and $y$ axes.
\end{theorem}

Note that if these new measures are not particularly favorable, the formula
boils down to the $O(n\delta\lg\delta)$ time
complexity of Section~\ref{sec:taking-advantage-f}.

\section{Taking Advantage of Diagonals of Blocks}
\label{sec:diagonalization}

In this section we take advantage of the relative positions of the
points to obtain a different adaptive algorithm.
We will consider partitions of the point set into two subsets.
These partitions are induced by two lines which are parallel to the
axes and perpendicular each other.
A combination of optimal boxes of each of the subsets will lead to the
optimal box of the whole point set.

For any subset $A\subseteq P$, a {\em diagonalization} of $A$ is 
a partition $\{A_1,A_2\}$ of $A$ induced by two lines $\ell_1$ and
$\ell_2$, respectively parallel to axes $x$ and $y$,
so that the elements of $A_1$ and the elements of $A_2$ belong to
opposite quadrants with respect to the point $\ell_1\cap \ell_2$.
%
  Figure~\ref{fig:diagonalization} gives some example of diagonalization.
  In particular, assuming that the points of $A_1$ are to the left of
  the points of $A_2$, we call the diagonalization {\em bottom-up}
  (Figure~\ref{fig:diagonalization}a) if the elements of $A_1$ are
  all below the elements of $A_2$ and {\em top-down}
  (Figure~\ref{fig:diagonalization}b) otherwise.
%
Note that if $p_1,p_2,\ldots,p_m$ denote the elements of $A$ sorted by
$x$-coordinate, then any diagonalization of $A$ has the form
$\{\{p_1,\ldots,p_k\},\{p_{k+1},\ldots,p_m\}\}$ for some index
$k\in[1..m-1]$.

  \begin{figure}[h]
    \centering
    \includegraphics[width=9cm]{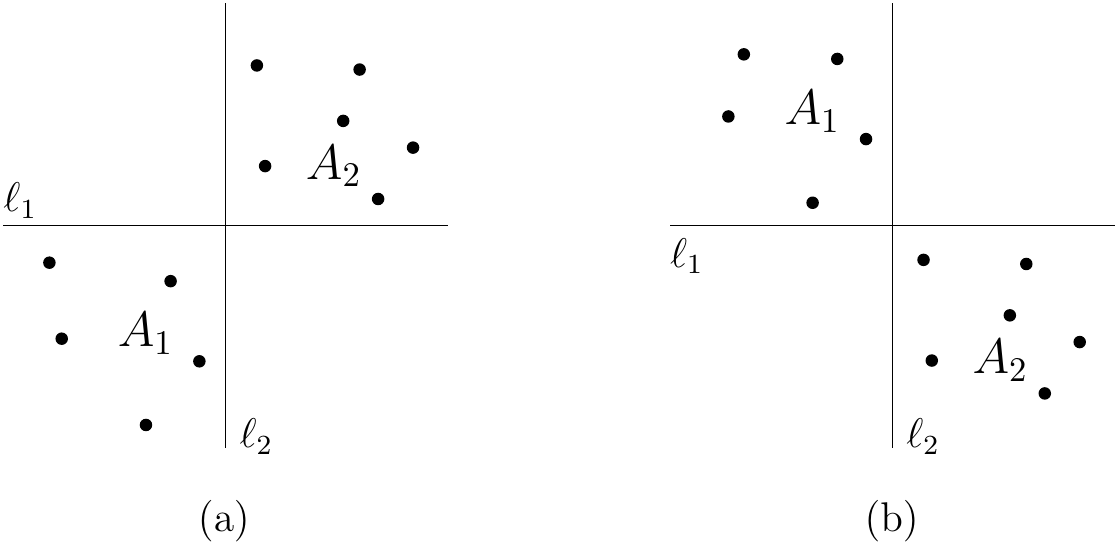}
    \caption{\small{A diagonalization $\{A_1,A_2\}$ of the point set
        $A=A_1\cup A_2$.}}
    \label{fig:diagonalization}
  \end{figure}

Given any bounded set $S\subset\mathbb{R}^2$, let $\bbox(S)$ denote
the smallest enclosing box of $S$ and let the {\em extreme} points of
$A$ be those belonging to the boundary of $\bbox(A)$.

Not all point sets can be decomposed into diagonals, the simplest case
being a set of four points placed at the four corners of a square
which sides are slightly rotated from the axes $x$ and $y$.
We call such a point set a \emph{windmill}, for the characteristic
position of its extreme points.
  See Figure~\ref{fig:windmills} for an example.

\begin{figure}[h]
  \centering
  \includegraphics[width=7cm]{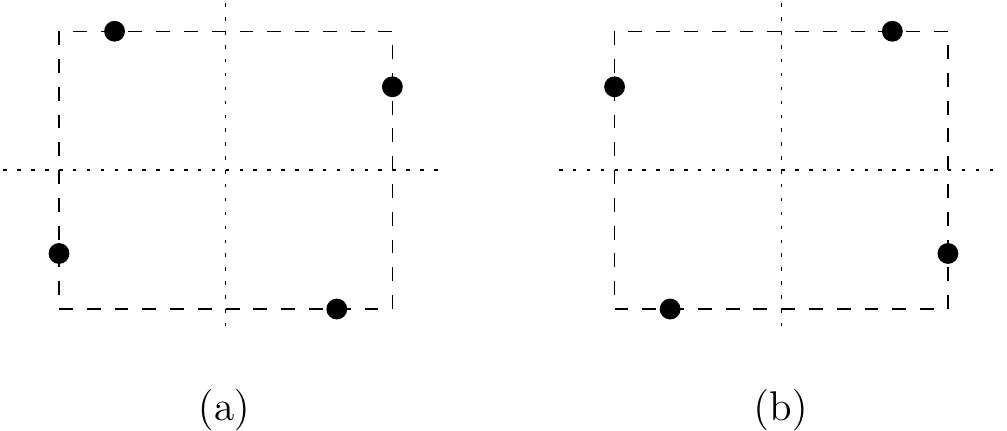}
  \caption{\small{Windmills.}}
  \label{fig:windmills}
\end{figure}

\begin{lemma}\label{lem:blocking-windmill}
  Let $A$ be a point set that does not admit a diagonalization. Then
  $A$ has exactly four extreme points. Furthermore, $A$ has a windmill
  which contains at least one extreme point of $A$.
\end{lemma}

\begin{proof}
The first part of the lemma is easy to show. We proceed to prove the second part.
Let $\{p,q,r,s\}$ be the extreme points of $A$. If they form a windmill
then we are done. Otherwise, assume without loss of generality that their relative
positions are as depicted in Figure~\ref{fig:not-diag-implies-windmill}a. 
\begin{figure}[h]
    \centering
    \includegraphics[width=10cm]{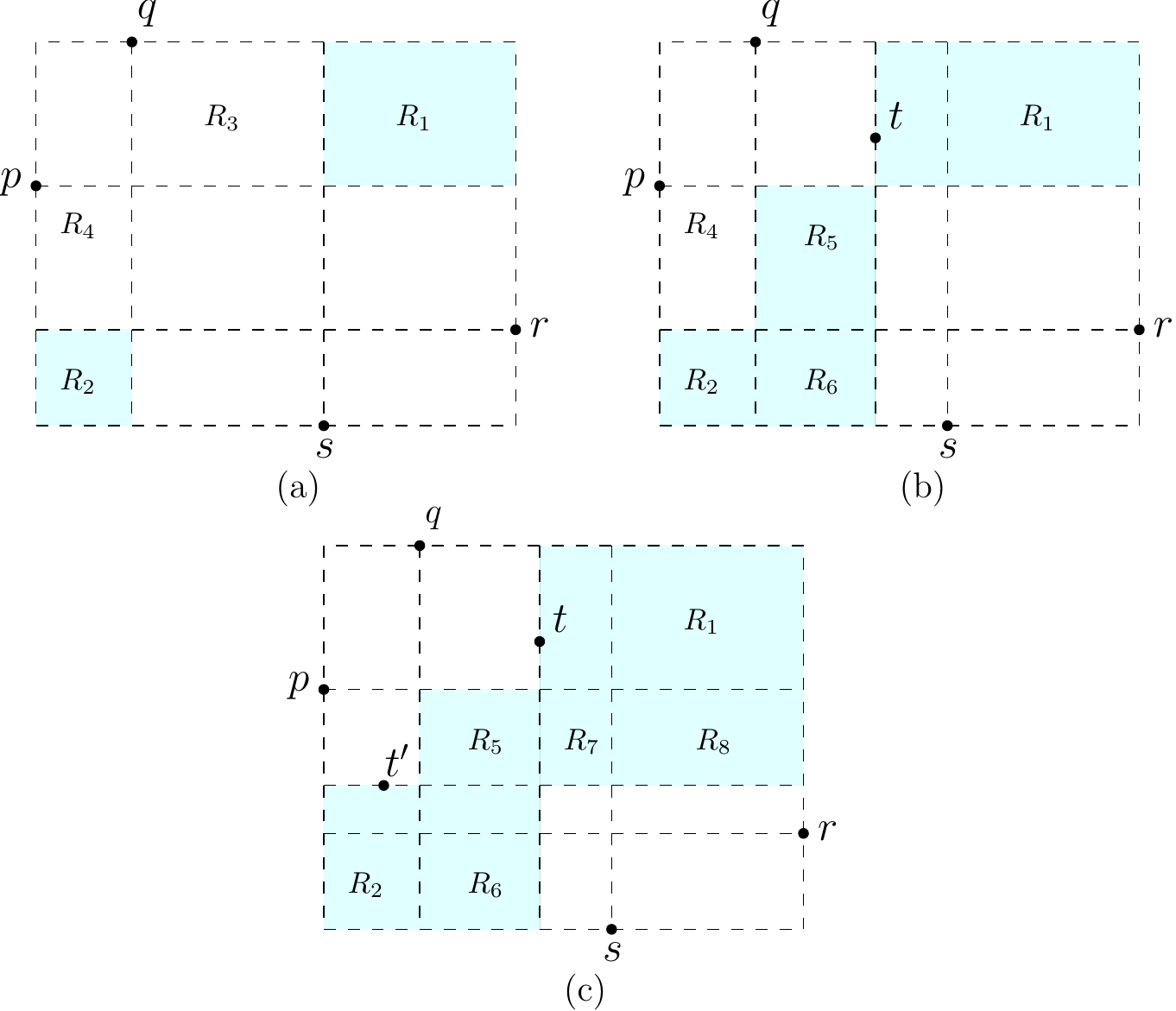}
    \caption{\small{Proof of Lemma~\ref{lem:blocking-windmill}. In each figure
    the shaded area is empty of points of set $A$.}}
    \label{fig:not-diag-implies-windmill}
\end{figure}
If there is an
element $q'$ of $A$ in region $R_1$ then $\{p,q',r,s\}$ is a windmill and we are
done. Analogously,
if there is a point $s'$ of $A$ in region $R_2$ then $\{p,q,r,s'\}$ is a windmill
and we are done. Assume then that $R_1\cup R_2$ is empty of elements of $A$. Since
$A$ does not admit a diagonalization then $R_3\cup R_4$ must contain a point of $A$.
Assume without loss of generality that $R_3$ contains a point. Let $t$ be the rightmost
point of $R_3$ creating regions $R_5$ and $R_6$ as shown in
Figure~\ref{fig:not-diag-implies-windmill}b. If there
is a point $u$ in $R_5$ then set $\{p,q,t,u\}$ is a windmill and we are done.
Analogously, if there
is a point $u$ in $R_6$ then set $\{p,t,r,u\}$ is a windmill and we are done. Then assume
$R_5\cup R_6$ is empty of elements of $A$.
Now, since $A$ does not admit a diagonalization, region $R_4$ must contain a point of
$A$. Let then $t'$ be the bottom-most point of $R_4$, which induces the regions $R_6$ and
$R_7$ depicted in
Figure~\ref{fig:not-diag-implies-windmill}c.
We now proceed with $t'$, similar as we did with $t$. 
If region $R_7$ contains a point $u$ of $A$ then $\{p,t,u,t'\}$ is a windmill and
the result follows. Otherwise,
region $R_8$ must contain a point $u'$ of $A$, because in the contrary case $A$ would
have a diagonalization. Then $\{t,u',s,t'\}$ is a windmill. 
Therefore, the result follows.
\end{proof}

\begin{definition}\label{defn:diagonalisationTree}
	A diagonalization tree of $P$,
  $D$-tree, is a binary tree such that: $(i)$ each leaf $u$ contains a
  subset $S(u)\subseteq P$ which does not admit a diagonalization,
  $(ii)$ set $\{S(u)~|~u\text{ is a leaf }\}$ is a partition of $P$,
  and $(iii)$ each internal node $v$ has exactly two children $v_1$
  (the left one) and $v_2$ (the right one) and satisfies that
  $\{A(v_1),A(v_2)\}$ is a diagonalization of $A(v)$, where for each
  node $v$ $A(v)$ denotes the union of the sets $S(u)$ for all leaves
  $u$ descendant of $v$ (See Figure~\ref{fig:diag-tree}).
\end{definition}

\begin{figure}[h]
    \centering
    \includegraphics[width=12cm]{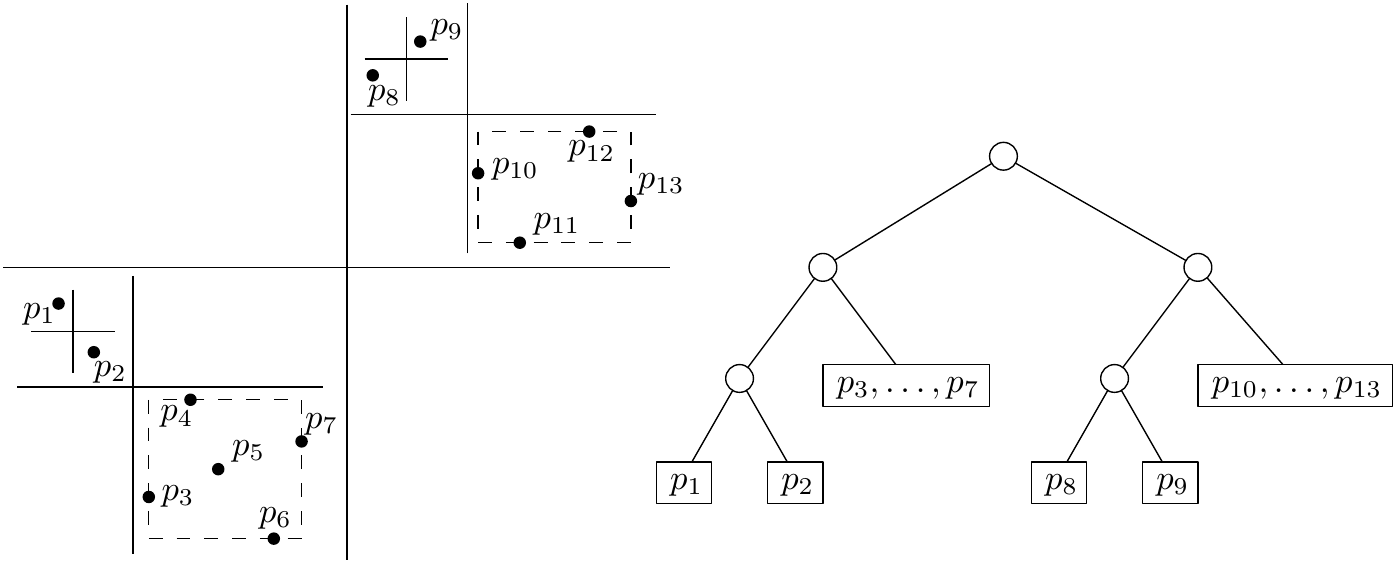}
    \caption{\small{A $D$-tree of the point set $\{p_1,\ldots,p_{13}\}$.}}
    \label{fig:diag-tree}
\end{figure}

\begin{lemma}\label{lem:number-of-nodes-diag-tree}
  Let $P$ be a set of $n$ points in the plane. Every $D$-tree of $P$ has the same
  number of leaves. Furthermore, the $i$-th leaves from left to right of any two
  $D$-trees of $P$ contain the same subset $S(\cdot)$ of $P$.
\end{lemma}

\begin{proof}
  We use induction on the number of elements of $P$. 
  If $P$ contains only one element, or $P$ does not admit a diagonalization, then we are done.
  Otherwise, assume that $P$ has exactly $k\geq 1$ different diagonalization, all of which must be of the same type (bottom-up or top-down).
  Then $P$ is partitioned into $k+1$ non-empty sets denoted from left to right $P_1,P_2,\ldots,P_{k+1}$ such that the $k$ diagonalizations are $\{\bigcup_{i=1}^jP_i,\bigcup_{i=j+1}^{k+1}P_{i}\}$ for $j=1\ldots k$. 
  Observe then that no set $P_i$ admits a bottom-up diagonalization and that any $D$-tree of $P$ can be obtained by building a binary tree, whose leaves from left to right are $P_1,P_2,\ldots,P_{k+1}$ and each internal node has two children, and replacing every leaf $P_i$ by a $D$-tree of $P_i$. 
  Applying the induction hypothesis for $P_1,P_2,\ldots,P_{k+1}$ the result follows.
\end{proof}

From Lemma~\ref{lem:number-of-nodes-diag-tree} we can conclude that every $D$-tree $T$ of
$P$ induces the same partition of $P$, which is equal to $\{S(u_1),\ldots,S(u_{\beta})\}$,
where $u_1,\ldots,u_{\beta}$ are the leaf nodes of $T.$


\begin{lemma}\label{lem:build-diag-tree}
  Let $P$ be a set of $n$ points in the plane. A $D$-tree of $P$ requires $O(n)$ space and
  can be built in $O(n\lg n)$ comparisons.
\end{lemma}

\begin{proof}\label{lem:build-diag-tree}
  Let $p_1,p_2,\ldots,p_n$ be the elements of $P$ sorted
  by $x$-coordinate, and let $p_{\pi_1},p_{\pi_2},\ldots,$ $p_{\pi_n}$ be the elements
  of $P$ sorted by $y$-coordinate. Both orders can be obtained in $O(n\lg n)$ coordinate comparisons.
  
  Considering the computation of permutation $\pi$ as a preprocessing, we now show that: If $P$ admits a diagonalization
  $\{\{p_1,\ldots,p_k\},\{p_{k+1},\ldots,p_n\}\}$ 
  then it can be determined in $O(\min\{k,n-k\})$ comparisons. Otherwise, if $P$ does
  not admit a diagonalization, then it can be decided in $O(n)$ comparisons.
  
  For each index $i\in[1..n]$, let $M_L(i)=\max_{j\in[1..i]}\pi_j$,
  $m_L(i)=\min_{j\in[1..i]}\pi_j$, $M_R(i)=\max_{j\in[i..n]}\pi_j$, and 
  $m_R(i)=\min_{j\in[i..n]}\pi_j$.
  
  Observe that if $\{\{p_1,\ldots,p_k\},\{p_{k+1},\ldots,p_n\}\}$ is a diagonalization of $P$, 
  then index $k\in[1..n-1]$ satisfies $M_L(k)=k$ or $m_L(k)=n-k+1$. Furthermore, $M_L(k)=k$ 
  and $m_L(k)=n-k+1$ are 
  equivalent to $m_R(k+1)=k+1$ and $M_R(k+1)=n-k$, respectively.
  
  Then we can determine a diagonalization of $P$, if it exists, as follows: 
  For $j=1..\lfloor n/2\rfloor$ decide if $\{\{p_1,\ldots,p_j\},\{p_{j+1},\ldots,p_n\}\}$ is
  a diagonalization (i.e.\ $M_L(j)=j$ or $m_L(j)=n-j+1$) or $\{\{p_1,\ldots,p_{n-j}\},\{p_{n-j+1},\ldots,p_n\}\}$ is a diagonalization (i.e.\ $M_R(n-j+1)=j$ or $m_R(n-j+1)=n-j+1$).
  Note that if $j>1$ then $M_L(j)$, $m_L(j)$, $M_R(n-j+1)$, and $m_R(n-j+1)$ can all be computed in $O(1)$ comparisons from $M_L(j-1)$, $m_L(j-1)$, $\pi_{j}$, $M_R(n-j+2)$, $m_R(n-j+2)$, and $\pi_{n-j+1}$.
  Therefore, if there is a diagonalization $\{\{p_1,\ldots,p_{k}\},\{p_{k+1},\ldots,p_n\}\}$ of $P$
  it is decided for $j=\min\{k,n-k\}\leq\lfloor n/2\rfloor$, and thus determined in $O(j)$ comparisons.
  If no diagonalization is found for any value of $j\in[1..\lfloor n/2\rfloor]$, then the algorithm spends $O(n)$ comparisons in total.
  
  We can then build a $D$-tree of $P$ recursively as follows. 
  Run the above algorithm for $P$.
  If a diagonalization $\{\{p_1,\ldots,p_k\},\{p_{k+1},\ldots,p_n\}\}$ of $P$ exists, which was determined
  in $O(t)$ comparisons where $t=\min\{k,n-k\}$, then create a root node and set as left child a 
  $D$-tree of $\{p_1,\ldots,p_k\}$ and as right child a $D$-tree of $\{p_{k+1},\ldots,p_n\}$. Otherwise, if $P$ does not admit a diagonalization, which was decided in $O(n)$ comparisons, then create a leaf node whose set $S(\cdot)$ is equal to $P$. This results in the next recurrence 
  equation for the total number $T(n)$ of comparisons, where
  $1\leq t\leq\lfloor n/2\rfloor$:
  \begin{equation*}
    T(n)=\left\{
      \begin{array}{ccc}
        O(t)+T(t)+T(n-t) & \text{    } & n>1,\text{ a diagonalization exists}\\
        O(n) & & \text{otherwise.}
      \end{array}
    \right.
  \end{equation*}
  W.l.o.g.\ assume that the constants in $O(t)$ and $O(n)$ in the recurrence
  are equal to one. Then we prove by induction that $T(n) \le n + n \lg n$. The base case of the induction is the second line of the recurrence equation, where $n \le n+n\lg n$ always holds. In the inductive case, we have:
  \begin{eqnarray*}
    T(n) &  =  & t + T(t) + T(n-t) \\
    & \le & n + t + t \lg t + (n-t) \lg (n-t) \\
    &  =  & n + t + t \lg t + (n-t) \lg (n-t) + (n\lg n-t\lg n-(n-t)\lg n) \\
    &  =  & n + n\lg n + t - t\lg(n/t) - (n-t)\lg(n/(n-t)) \\
    &  =  & n + n\lg n + n(t/n - H(t/n))\\
    & \le & n + n\lg n
  \end{eqnarray*}
  The second line uses the inductive hypothesis. In the third line we add and
  subtract $n\lg n$ written in two different forms. In the fourth line we regroup terms. In the fifth line we introduce the binary entropy function
  $H(x) = x\lg(1/x) + (1-x)\lg(1/(1-x))$, where $x=t/n$. Finally, in the last
  line we apply the analytic inequality $x \le H(x)$, which holds at least for $x \le 1/2$. 
  Thus $T(n) \le n + n\lg n$ and then $T(n)$ is $O(n\lg n)$.
  One can see that this solution is tight by considering the case $t=n/2$.
  
  It is easy to see that any $D$-tree
  of $P$ requires $O(n)$ space. The result follows.
\end{proof}

\begin{definition}\label{defn:ten}
  For any non-empty subset $A\subseteq P$ 
  the set of ``ten'' $f()$-optimal boxes of $A$, denoted by $\ten(A)$, consists of the
  following $f()$-optimal boxes of $A$, all contained in $\bbox(A)$:
  \begin{enumerate}
  \item $\bbox(A)$.
  \item $\xbox_{opt}(A)$, an $f()$-optimal box.
  \item $\xbox_1(A)$, an $f()$-optimal box containing the bottom-left vertex of
    $\bbox(A)$.
  \item $\xbox_2(A)$, an $f()$-optimal box containing the bottom-right vertex of
    $\bbox(A)$.
  \item $\xbox_3(A)$, an $f()$-optimal box containing the top-right vertex of
    $\bbox(A)$.
  \item $\xbox_4(A)$, an $f()$-optimal box containing the top-left vertex of
    $\bbox(A)$.
  \item $\xbox_{1,2}(A)$, an $f()$-optimal box containing the bottom vertices
    of $\bbox(A)$.   
  \item $\xbox_{2,3}(A)$, an $f()$-optimal box containing the right vertices
    of $\bbox(A)$. 
  \item $\xbox_{3,4}(A)$, an $f()$-optimal box containing the top vertices
    of $\bbox(A)$. 
  \item $\xbox_{4,1}(A)$, an $f()$-optimal box containing the left vertices
    of $\bbox(A)$. 
  \end{enumerate}
\end{definition}

\begin{lemma}\label{lem:compute-ten}
  For any non-empty subset $A\subseteq P$ and any diagonalization $\{A_1,A_2\}$ of $A$,
  $\ten(A)$ can be computed in $O(1)$ score compositions from $\ten(A_1)$ and
  $\ten(A_2)$.
\end{lemma}

\begin{proof}
  Suppose without loss of generality that $\{A_1,A_2\}$
  is bottom-up \begin{LONG}(see Figure~\ref{fig:diagonalization}a)\end{LONG}.
  Let $u_2$ and $u_4$ be the bottom-right and top-left
  vertices of $\bbox(A)$, respectively. 
  Let $\max\{H_1,\ldots,H_t\}$ denote the $f()$-optimal box for $A$ among the boxes $H_1,\ldots,H_t$.  
  Observe that:
  
  \begin{minipage}{1.0\linewidth}
    \begin{itemize}
    \item $\bbox(A)=\bbox(A_1\cup A_2)=\bbox(\bbox(A_1)\cup\bbox(A_2))$.
    \item $\xbox_{opt}(A)=\max\{\xbox_{opt}(A_1),\xbox_{opt}(A_2), \bbox(\xbox_3(A_1)\cup\xbox_1(A_2))\}$.
    \item $\xbox_1(A)=\max\{\xbox_{1}(A_1),\bbox(\bbox(A_1)\cup\xbox_1(A_2))\}$.
    %
    %
    \item $\xbox_2(A)=\bbox(\max\{\xbox_2(A_1),\xbox_2(A_2), \bbox(\xbox_{2,3}(A_1)\cup\xbox_{1,2}(A_2))\} \cup\{u_2\})$.
    \item $\xbox_3(A)=\max\{\bbox(\xbox_3(A_1)\cup\bbox(A_2)),\xbox_{3}(A_2)\}$.
    %
    %
    \item $\xbox_4(A)=\bbox(\max\{\xbox_4(A_1),\xbox_4(A_2), \bbox(\xbox_{3,4}(A_1)\cup\xbox_{4,1}(A_2))\}\cup\{u_4\})$.
    \item $\xbox_{1,2}(A)=\max\{\bbox(\xbox_{1,2}(A_1)\cup \{u_2\}),\bbox(\bbox(A_1)\cup\xbox_{1,2}(A_2))\}$
    \item $\xbox_{2,3}(A)=\max\{\bbox(\xbox_{2,3}(A_2)\cup \{u_2\}),\bbox(\bbox(A_2)\cup\xbox_{2,3}(A_1))\}$
    \item $\xbox_{3,4}(A)=\max\{\bbox(\xbox_{3,4}(A_2)\cup \{u_4\}),\bbox(\bbox(A_2)\cup\xbox_{3,4}(A_1))\}$
    \item $\xbox_{4,1}(A)=\max\{\bbox(\xbox_{4,1}(A_1)\cup \{u_4\}),\bbox(\bbox(A_1)\cup\xbox_{4,1}(A_2))\}$
    \end{itemize}
  \end{minipage}
  
  Since each of the elements of $\ten(A)$ can be obtained from a constant
  number of pairwise disjoint boxes of $\ten(A_1)\cup\ten(A_2)$, and for any two
  disjoint boxes $H$ and $H'$ we have $f((H\cup H')\cap A)=f((H\cap A)\cup(H'\cap
  A))=g(f(H\cap A),f(H'\cap A))$, the result follows.
\end{proof}

\begin{theorem}\label{teo:comp-optimal-box-diag-tree}
  Let $P$ be a set of $n$ points in the plane.
  Let $f()$ be a monotone decomposable function receiving as argument
  any subset of $P$.
  There exists an algorithm that finds an $f()$-optimal box of $P$ in
  $O(n\lg n+\sum_{i=1}^{\beta}h_c(n_i))$ comparisons (on coordinates
  and indices) and $O(\sum_{i=1}^{\beta}h_s(n_i)+\beta)$ score
  compositions,
  where $\{P_1,\ldots,P_{\beta}\}$ is the partition of $P$ induced by
  any $D$-tree of $P$ and $\beta$ is the size of this partition,   
  $n_i$ is the cardinality of $P_i$, and
  $h_c(n_i)$ and $h_s(n_i)$ are the numbers of coordinate comparisons
  and score compositions used, respectively, to compute the ``ten''
  $f()$-optimal boxes of $P_i$.
\end{theorem}

\begin{proof}
  Build a $D$-tree $T$ of $P$ in $O(n\lg n)$ comparisons
  (Lemma~\ref{lem:build-diag-tree}).
  Let $u_1,\ldots,u_{\beta}$ be the leaves of $T$ which satisfy
  $S(u_i)=P_i$ for all $i\in[1..n]$. Compute the set $\ten(S(u_i))=\ten(P_i)$ in $h_c(n_i)$
  coordinate comparisons and $h_s(n_i)$ 
  score compositions. 
  By using a post-order traversal of $T$, for each internal node $v$
  of $T$ compute $\ten(A(v))$ from $\ten(A(v_1))$ and $\ten(A(v_2))$,
  where $v_1$ and $v_2$ are the children nodes of $v$, in $O(1)$ score
  compositions (Lemma~\ref{lem:compute-ten}). 
  The $f()$-optimal box of $P$ is the box $\xbox_{opt}(A(r))$, where
  $r$ is the root node of $T$ and satisfies $A(r)=P$.
  In total, this algorithm runs in
  $O(n\lg n)+\sum_{i=1}^{\beta}h_c(n_i)=O(n\lg
  n+\sum_{i=1}^{\beta}h_c(n_i))$ coordinate comparisons and
  $\sum_{i=1}^{\beta}h_s(n_i)+\sum_{i=1}^{\beta-1}O(1)=O(\sum_{i=1}^{\beta}h_s(n_i)+\beta)$
  score compositions.
\end{proof}

Observe that the best case of the algorithm of
Theorem~\ref{teo:comp-optimal-box-diag-tree} is when $|S(u_i)|$ is
$O(1)$ for each leaf node $u_i$ of the $D$-tree of $P$.
It yields a complexity of $O(n\lg n)$ coordinate comparisons and
$O(n)$ score compositions.
\begin{LONG}
  Another observation is that when $S(u_1),\ldots,S(u_{\beta})$ are of
  similar sizes, the time complexity increases to $O(n\lg n)$
  coordinate comparisons and $O((n^2/\beta)\lg (n/\beta)+\beta)$ score
  compositions in the worst case.
\end{LONG}
The worst case is when $P$ does not admit a diagonalization, but even
in this case the additional $O(n\lg n)$ coordinate comparisons,
performed by the algorithm while trying to diagonalize, do not add to
the overall asymptotic complexity.
Hence the worst case complexity of the diagonalization algorithm
reduces to the case studied in the previous sections, with a
complexity of at most $O(n\lg n)$ coordinate comparisons and $O(n^2\lg
n)$ score compositions.

\section{Dealing with Windmills}
\label{subsec:break-windmills}

In this section we use Lemma~\ref{lem:blocking-windmill} to obtain a variant of the algorithm
in Theorem~\ref{teo:comp-optimal-box-diag-tree}. The set $S(u)$ of every leaf node $u$ of any 
$D$-tree of $P$ does not admit a diagonalization and has a windmill containing an extreme point
of $S(u)$. The idea is to remove the extreme points of $S(u)$ and then recursively build
a $D$-tree of the remaining points. 
This approach yields a diagonalization in depth of the point set,
potentially reducing the number of score compositions.

\begin{definition}\label{def:ex-diag-tree}
  An extended diagonalization tree of $P$, $D^*$-tree, is defined recursively as
  follows: 
  Each leaf node $u$ of a $D$-tree of $P$ satisfying $|S(u)|>1$ is
  replaced by a node $u'$ containing the set $X(u)$ of the four
  extreme points of $S(u)$, and if the set $S(u)\setminus X(u)$ is not
  empty then $u'$ has as its only one child a $D^*$-tree of
  $S(u)\setminus X(u)$.  
\end{definition}

\begin{lemma}\label{lem:build-ex-diag-tree}
  Let $P$ be a set of $n$ points in the plane. 
  Every $D^*$-tree of $P$ has the same number $\sigma$ of one-child
  nodes, contains $n-4\sigma$ leaves nodes, and every leaf node $u$
  satisfies $|S(u)|=1$ or $|S(u)|=4$.
  A $D^*$-tree of $P$ requires $O(n)$ space and can be built in
  $O(n\lg n+\sigma n)$ comparisons.  
\end{lemma}
\begin{proof} 
  The first part of the proof can be seen from
  Lemma~\ref{lem:number-of-nodes-diag-tree} and
  Definition~\ref{def:ex-diag-tree}.
  A $D^*$-tree of $P$ can be built in $O(n\lg n+\sigma n)$ comparisons
  by following the same algorithm to build a $D$-tree of $P$ until
  finding a leaf node $u$ such that $S(u)$ does not admit a
  diagonalization.
  At this point we pay $O(n)$ comparisons in order to continue the
  algorithm with the set $S(u)\setminus X(u)$ according to
  Definition~\ref{def:ex-diag-tree}.
  Since this algorithm finds $\sigma$ nodes $u$, the total comparisons
  are $O(n\lg n+\sigma n)$.
  The $D^*$-tree has $n$ nodes of bounded degree and hence can be
  encoded in linear space.  
\end{proof}

\begin{theorem}\label{teo:comp-optimal-box-windmills}
  Let $P$ be a set of $n$ points in the plane.
  Let $f()$ be a monotone decomposable function receiving as argument
  any subset of $P$.
  There exists an algorithm that finds an $f()$-optimal box of $P$ in
  $O(n\lg n+\sigma n)$ coordinate comparisons and
  $O(n+\sigma n\lg n)$ score
  compositions, where $\sigma$ is the number of one-child nodes of every 
  $D^*$-tree of $P$.
\end{theorem}

\begin{proof}
  Build a $D^*$-tree $T$ of $P$ in $O(n\lg n+\sigma n)$ 
  comparisons (Lemma~\ref{lem:build-ex-diag-tree}).
  For each of the $n-4\sigma$ leaves nodes $u$ of $T$ compute $\ten(S(u))$ in constant
  score compositions.
  Then, using a post-order traversal of $T$, compute $\ten(S(u))$ for each internal node $u$ as follows:
  If $v$ has two children $v_1$ (the left one) and $v_2$ (the right one), then $\{A(v_1),A(v_2)\}$
  is a diagonalization of $A(v)$ and $\ten(A(v))$ can be computed in $O(1)$ score compositions
  from $\ten(A(v_1))$ and $\ten(A(v_2))$ (Lemma~\ref{lem:compute-ten}).
  Otherwise, if $v$ is one of the $\sigma$ one-child nodes, then $\ten(A(v))$ can be computed
  in $O(n\lg n)$ worst-case comparisons and score compositions. Namely, if a box
  of $\ten(A(v))$ contains a at least one point of $X(u)$ in the boundary then it
  can be found in $O(n\lg n)$ comparisons and score compositions~\cite{cortes2009}. Otherwise,
  it is a box of $\ten(A(v'))$, where $v'$ is the child of $v$.
  We pay $O(1)$ score compositions for each of the $O(n)$ two-child nodes and $O(n\lg n)$ score
  compositions for each of the $\sigma$ one-child nodes. Then the total score compositions
  is $O(n+\sigma n\lg n)$. The result follows.
\end{proof}

\section{Conclusions}
\label{sec:conclusion}

Cort\'es et al.~\cite{cortes2009} proposed a solution for the
\pb{Maximum Weight Box} based on a data structure maintaining the
scores of a hierarchy of segments, the \emph{MCS Tree}.
We extended this solution in several directions:
\begin{enumerate}
\item we showed how to replace the sum operator by a monotone
  decomposable function $f()$, so that Cort\'es et al.'s
  algorithm~\cite{cortes2009} can optimize more general score
  functions than merely the sum of the weights of the points;
  \begin{LONG}
  \item we refined the complexity analysis of such algorithms,
    separating the amount of coordinate comparisons and the number of
    applications of the composition function $g()$ (replacing the
    addition of reals in Cort\'es et al.'s algorithm), accounting for
    cases where $g()$ is computed in more than constant time;
  \end{LONG}
\item we extended the \emph{MCS tree} data structure to make it fully
  dynamic, supporting the insertion and removal of points in time
  logarithmic in the number of points in the structure and in their
  relative insertion ranks;
\item we described adaptive techniques to take advantage of various
  particularities of instances, such as the clustering of positive and
  negative points or the relative positions of the points, without
  sacrificing worst case complexity.
\end{enumerate}

Whereas we showed many techniques to take advantage of particular
instances of the \pb{Optimal Planar Box}, other types of instances are
left to study, such as more complex forms of clustering, and a unified
analysis of all the adaptive techniques presented.
Other directions of research are the generalization of the problem to
higher dimension, and to other shapes than boxes, such as convex
polygons.

\bibliographystyle{plain}
\bibliography{paper}

\begin{thebibliography}{10}

\bibitem{avl1962}
G.~Adelson-Velskii and E.~M. Landis.
\newblock An algorithm for the organization of information.
\newblock In {\em Proceedings of the USSR Academy of Sciences}, volume 146,
  pages 263--266, 1962.
\newblock (Russian) English translation by M. J. Ricci in Soviet Math. Doklady,
  3:1259-1263, 1962.

\bibitem{bautista2011}
C.~Bautista-Santiago, J.~M. D\'{\i}az-B{\'a}{\~n}ez, D.~Lara,
  P.~P{\'e}rez-Lantero, J.~Urrutia, and I.~Ventura.
\newblock Computing optimal islands.
\newblock {\em Oper. Res. Lett.}, 39(4):246--251, 2011.

\bibitem{bentley1984}
J.~Bentley.
\newblock Programming pearls: algorithm design techniques.
\newblock {\em Commun. ACM}, 27(9):865--873, 1984.

\bibitem{cole2000-1}
R.~Cole, B.~Mishra, J.~Schmidt, and A.~Siegel.
\newblock On the dynamic finger conjecture for splay trees. {P}art {I}: {S}play
  sorting $\log n$-block sequences.
\newblock {\em SIAM J. Comp.}, 30(1):1--43, 2000.

\bibitem{cormen2001}
T.~H. Cormen, C.~Stein, R.~L. Rivest, and C.~E. Leiserson.
\newblock {\em Introduction to Algorithms}.
\newblock McGraw-Hill Higher Education, 2nd edition, 2001.

\bibitem{cortes2009}
C.~Cort\'{e}s, J.~M. D\'{\i}az-B\'{a}\~{n}ez, P.~P\'{e}rez-Lantero, C.~Seara,
  J.~Urrutia, and I.~Ventura.
\newblock Bichromatic separability with two boxes: A general approach.
\newblock {\em J. Algorithms}, 64(2-3):79--88, 2009.

\bibitem{dobkin1996}
D.~P. Dobkin, D.~Gunopulos, and W.~Maass.
\newblock Computing the maximum bichromatic discrepancy, with applications to
  computer graphics and machine learning.
\newblock {\em J. Comput. Syst. Sci.}, 52(3):453--470, 1996.

\bibitem{dobkin1989}
D.~P. Dobkin and S.~Suri.
\newblock Dynamically computing the maxima of decomposable functions, with
  applications.
\newblock In {\em FOCS}, pages 488--493, 1989.

\bibitem{eckstein2002}
J.~Eckstein, P.~Hammer, Y.~Liu, M.~Nediak, and B.~Simeone.
\newblock The maximum box problem and its application to data analysis.
\newblock {\em Comput. Optim. App.}, 23(3):285--298, 2002.

\bibitem{estivillcastro92survey}
V.~Estivill-Castro and D.~Wood.
\newblock A survey of adaptive sorting algorithms.
\newblock {\em ACM Comp. Surv.}, 24(4):441--476, 1992.

\bibitem{theArtOfComputerProgramming}
D.E. Knuth.
\newblock {\em The Art of Computer Programming}, volume~3.
\newblock Addison-Wesley, 1968.

\bibitem{liu2003}
Y.~Liu and M.~Nediak.
\newblock Planar case of the maximum box and related problems.
\newblock In {\em CCCG}, pages 14--18, 2003.

\bibitem{measuresOfPresortednessAndOptimalSortingAlgorithms}
Heikki Mannila.
\newblock Measures of presortedness and optimal sorting algorithms.
\newblock In {\em IEEE Trans. Comput.}, volume~34, pages 318--325, 1985.

\bibitem{anOverviewOfAdaptiveSorting}
A.~Moffat and O.~Petersson.
\newblock An overview of adaptive sorting.
\newblock {\em Australian Comp. J.}, 24(2):70--77, 1992.

\bibitem{sleator1985}
D.~D. Sleator and R.~E. Tarjan.
\newblock Self-adjusting binary search trees.
\newblock {\em J. ACM}, 32(3):652--686, 1985.

\bibitem{takaoka2002}
T.~Takaoka.
\newblock Efficient algorithms for the maximum subarray problem by distance
  matrix multiplication.
\newblock {\em Electronic Notes in Theoretical Computer Science}, 61:191--200,
  2002.
\newblock CATS'02, Computing: the Australasian Theory Symposium.

\end{thebibliography}

\end{document}